\documentclass[11pt]{article}
\usepackage{a4wide}

\usepackage[utf8]{inputenc}
\usepackage[T1]{fontenc}
\usepackage{graphicx}
\usepackage{longtable}
\usepackage{wrapfig}
\usepackage{rotating}
\usepackage[normalem]{ulem}
\usepackage{amsmath}
\usepackage{amssymb}
\usepackage{capt-of}
\usepackage{hyperref}
\usepackage{color}
\usepackage{listings}
\usepackage{amsmath}
\usepackage{amsthm}
\usepackage{physics}
\usepackage{dsfont}
\usepackage{bm}

\newtheorem{theorem}{Theorem}
\newtheorem{lemma}{Lemma}
\newtheorem{definition}{Definition}

\newcommand{\one}{\mathds{1}}

\title{Emergence of Objectivity for Quantum Many-Body Systems}
\author{Harold Ollivier} 
\date{INRIA, 2 rue Simone Iff, F-75012 Paris, France. \\[2ex] \today}

\begin{document}

\maketitle
\begin{abstract}
  We examine the emergence of objectivity for quantum
  many-body systems in a setting without an environment to decohere the
  system's state, but where observers can only access small fragments
  of the whole system. We extend the result of Reidel (2017)
  to  the case where the system is in a mixed state, measurements are
  performed through POVMs, and imprints of the outcomes are
  imperfect. We introduce a new condition on states and measurements
  to recover full classicality for any number of observers. We further
  show that evolutions of quantum many-body systems can be expected to
  yield states that satisfy this condition whenever the 
  corresponding measurement outcomes are redundant.
\end{abstract}

\section{Introduction}
\label{sec:org3d293ae}

{The emergence of classical reality from within a quantum mechanical
universe has always been central to discussions on the foundations of
quantum theory.  While decoherence---through interactions of a
quantum system with its environment---accounts for the disappearance
of superpositions of quantum
states \cite{Z81:pointer,Z82:environment,JZ85}, it does not provide an
{a priori} explanation for all intrinsic properties of a
classical world and, in particular, for the emergence of an objective
classical reality.}

Quantum Darwinism 
\cite{OPZ04:objective, OPZ05:environment,
  BZ05:simple, BZ06:quantum, BZ08:quantum, Z09:quantum, Z14:quantum,
  Z18:quantum} proposes a solution to fill this gap. Its {credo}
states that rather than interacting directly with systems of interest,
observers intercept a small fraction of their environment to gather
information about them.  Classicality then emerges naturally from
quantum Darwinism.  First, observing the system of interest \(\mathcal
S\) indirectly, by measuring its environment \(\mathcal E\) rather
than directly with an apparatus, restricts obtainable information to
observables on \(\mathcal S\) that are faithfully recorded in the
environment. In practice, these observables are commuting with the
well-defined preferred pointer basis induced by decoherence due to the
interaction Hamiltonian between \(\mathcal S\) and \(\mathcal
E\). Second, requiring the observer to be able to infer the state of
\(\mathcal S\) by measuring only a small fraction of \(\mathcal E\)
implies that many such observers can do the same without modifying the
state of the system.  This, in turn, grants the state of the system an
objective existence, as it can be discovered and agreed upon by many
observers.

While early descriptions of quantum
Darwinism \cite{OPZ04:objective,OPZ05:environment} focused on simple
models to build intuition, several subsequent works have studied the
redundancy of information in more complex settings.
References \cite{BZ08:quantum,RZ10:quantum,KHH14:objectivity,ZRZ14:amplification,ZRZ16:amplification}
show that quantum Darwinism---through the redundant proliferation of
information about the pointer states in the environment---is a rather
ubiquitous phenomenon encountered in many realistic situations.

The models used above to exemplify quantum Darwinism consider that the
whole universe can be naturally split between \(\mathcal S\), the system
of interest, and \(\mathcal E\), the environment itself subdivided into
subsystems \(\mathcal E = \cup_i \mathcal E_i\). As a consequence, the
emergence of classicality is {de facto} analyzed relative to this
separation. Redundant information is sought about observables on
\(\mathcal S\) in \(\mathcal E\).  Yet, this is already going beyond what
seems to be the minimal requirement that should allow to recover
classical features of the universe: a natural egalitarian
tensor--product structure for the state space, without explicit
reference to a preferred system--environment dichotomy.

Such a scenario is particularly relevant for the Consistent Histories 
framework \cite{G84:consistent,O94:interpretation,GO99:consistent}.
The universe is viewed as a closed quantum system in which one wants
to identify a single set of consistent histories that describe the
quasi-classical domain, where emergent coarse-grained observables
follow the classical equations of motion \cite{RZZ16:objective}, and
become objective for observers embedded in the quantum universe. In a
similar fashion, this scenario is adapted for understanding the
emergence of objective properties in many-body physics.  The reason is
that for such composite systems, quantum fluctuations can be recorded
into complex mesoscopic regions, e.g., in the course of their
amplification by classically chaotic systems.  Hence, redundant
information need not be relative to observables of a single subsystem
or any predefined set of subsystems, but rather to observables of
to-be-determined sets of subsystems. Ref. \cite{R17:classical}
examines this question and shows that, due to the absence of a fixed
set of subsystems defining the system of interest \(\mathcal S\), it is
possible to construct redundant records for two mutually incompatible
observables.  While this gives a clear example where redundancy of
information is not enough to guarantee the uniqueness of objective
observables, the main result of \cite{R17:classical} shows that this
ambiguity requires the redundant records to delicately overlap with
one another. In practical situations, such a delicate overlap is
expected to be unlikely, thereby recovering the usual uniqueness of
objective observables.

The present work shows that a similar conclusion can be expected in a
more general setting, where redundant records are not required to be
perfectly imprinted in the Hilbert space of the whole universe and
where observables are replaced with POVMs (Throughout this
paper, observable refer to sharp observables, so that POVMs are a
generalization of observables).   
 To this end, Section \ref{sec:ideal}
presents an overview of \cite{R17:classical} and outlines some of
the key ingredients used implicitly when relying on perfect redundant
records of observables. Section \ref{sec:generalize} generalizes
the tools defining redundancy and classicality to our 
scenario. Section \ref{sec:tuple} provides a sufficient criterion
on the approximate redundant records to recover classicality for a 
single set of POVMs on \(\mathcal S\). Finally,
Section \ref{sec:dynamical} takes a dynamical perspective to the
emergence of objectivity and shows that our criterion is expected to
hold in a wide range of situations, thereby implying that quantum
Darwinism is a ubiquitous explanation for the emergence of classical
properties in quantum many-body systems.

\section{Objectivity for Idealized Quantum Many-Body Systems \label{sec:ideal}}
\label{sec:org222f6a6}

In Ref. \cite{R17:classical}, an archetypal quantum many-body
system is introduced to study the emergence of objective properties.
It consists of a quantum system \(\mathcal S\) composed of a collection
of microscopic quantum systems \(\mathcal S = \cup_{i =1}^N \mathcal
S_i\). As a consequence, the Hilbert space \(\mathcal H_{\mathcal S}\) of
\(\mathcal S\) has a natural tensor--product structure, \(\mathcal
H_{\mathcal S} = \bigotimes_i {\mathcal H}_{\mathcal S_i}\).

Objective classical properties for \(\mathcal S\) are expected to emerge
from redundant imprints that are accessible to observers using 
feasible measurements on fractions of \(\mathcal S\). More precisely,
assuming \(\mathcal S\) is in a pure state \(\ket \psi\), redundant
observables should induce a decomposition of \(\ket \psi\) into
orthogonal but un-normalized branches \(\ket {\psi_i}\) 
\begin{equation}\label{eq:common_branch}
\ket \psi = \sum_i \ket {\psi_i},
\end{equation}
each \(\ket{\psi_i}\) being a common eigenstate of the redundant
observables. This implies that, for measurements on fractions of
\(\mathcal S\), this coherent superposition is indistinguishable from
the incoherent classical mixture \(\sum_i \ketbra{\psi_i}\), thus
forbidding observers to experience the quantumness of the correlations
between fragments of \(\mathcal S\).

The similarity with quantum Darwinism should be clearly apparent: for
both, not all subsystems can be measured simultaneously, thus forcing
partial observations. In the presence of faithful redundant imprints, this
would allow several observers to agree on their measurement results,
thereby granting those records and associated observables an objective
existence.

However, the similarity stops here. For quantum many-body systems, one
cannot readily conclude that evolutions inducing faithful redundant
imprints will favor the emergence of a \emph{single} set of redundant
observables, contrarily to usual system--environment
settings \cite{BPH15:generic}. The reason for such difference stems
from the absence, in the many-body setting, of precise localization
for the redundant records themselves.

For instance, in Ref. \cite{BPH15:generic}, although the choice
of one subsystem of the whole universe for playing the role of
reference system is arbitrary---any other would be equivalent for the
purpose of the conducted analysis---it is clearly identified, and the
redundant imprints refer to a measurement record of an observable for
this specific subsystem. Therefore, comparisons between the
conclusions drawn for different choices of the reference subsystem cannot
be made. Even more strikingly, Ref. \cite{R17:classical} gives a
concrete example of two redundantly recorded, yet non-commuting,
observables for \(\mathcal S\). One or the other could then equally
pretend to be objective, while their combination does not allow the
branch decomposition of Equation (\ref{eq:common_branch}).

To see this, consider \(\mathcal S\) made of qubits \(\mathcal S_{i,j}\)
where \((i,j) \in [1,N]\times [1,N]\).  The state of \(\mathcal S\) is
prepared by applying a CPTP map \(\Lambda\) from a single qubit to
\(\mathcal S\) and defined in the following way:
\begin{align*}
\ket{0}
& \rightarrow \ket{\bar 0} = \frac{1}{\sqrt{2^N}} \bigotimes_{i=1}^{N} \left( \bigotimes_{j=1}^{N} \ket{0}_{i,j} + \bigotimes_{j=1}^{N}\ket{1}_{i,j} \right) \\
\ket{1}
& \rightarrow \ket{\bar 1} = \frac{1}{\sqrt{2^N}} \bigotimes_{i=1}^{N} \left( \bigotimes_{j=1}^{N} \ket{0}_{i,j} - \bigotimes_{j=1}^{N}\ket{1}_{i,j} \right).
\end{align*}

Clearly, for fixed \(i\), the measurement of the qubits labeled
\(\{(i,j), j \in [1,N]\}\) in the basis \((\bigotimes_{j=1}^{N}
\ket{0}_{i,j} \pm \bigotimes_{j=1}^{N}\ket{1}_{i,j})/\sqrt{2}\) is 
equivalent to the measurement of the whole system relative to the basis 
\(\{\ket{\bar 0}, \ket{\bar 1}\}\).  This means that the information
about the observable \(\bar Z = \ketbra{\bar 0} - \ketbra{\bar 1}\) is
perfectly imprinted \(N\) times in \(\mathcal S\).

In addition, one can also rewrite the vectors \(\ket{\bar 0}\) and
\(\ket{\bar 1}\):
\begin{align}
\ket{\bar 0}
& = \frac{1}{\sqrt{2^N}} \bigotimes_{i=1}^N \left(\bigotimes_{j=1}^{N} \ket{0}_{i,j} + \bigotimes_{j=1}^{N}\ket{1}_{i,j} \right) \nonumber \\
& = \frac{1}{\sqrt{2^N}} \sum_{b=0}^{2^N -1} \bigotimes_{j=1}^{N} \ket{b_j} \label{eq:0_rewrite} \\
\ket{\bar 1}
& = \frac{1}{\sqrt{2^N}} \bigotimes_{i=1}^N \left(\bigotimes_{j=1}^{N} \ket{0}_{i,j} - \bigotimes_{j=1}^{N}\ket{1}_{i,j} \right) \nonumber \\
& = \frac{1}{\sqrt{2^N}} \sum_{b=0}^{2^N -1} \bigotimes_{j=1}^{N} (-1)^b \ket{b_j}, \label{eq:1_rewrite}
\end{align}
where, for a given \(b\) written as a binary string \(b = (b_1, \ldots,
b_{N})\), \(\ket{b} = \bigotimes_{i=1}^{N} \ket{b_i}_i\).  Combining
Equations (\ref{eq:0_rewrite}) and (\ref{eq:1_rewrite}), the conjugate
basis has a simple expression:
\begin{align*}
\frac{\ket{\bar 0} + \ket{\bar 1}}{\sqrt 2}
& = \frac{1}{\sqrt{2^N}} \sum_{\underset{h(b):\mathrm{even}}{b \in [0,2^N -1]}} \bigotimes_{j=1}^{N} \ket{b_j} \\
\frac{\ket{\bar 0} - \ket{\bar 1}}{\sqrt 2}
& = \frac{1}{\sqrt{2^N}} \sum_{\underset{h(b):\mathrm{odd}}{b \in [0,2^N -1]}} \bigotimes_{j=1}^{N} \ket{b_j},
\end{align*}
where \(h(b)\) denotes the Hamming weight of \(b\). As a consequence
of this rewrite, for fixed \(j\), any measurement of qubits labeled
\(\{(i,j), i \in [0,N]\}\) that reveals the parity of the weight of
\(b_j\) is equivalent to a measurement of the conjugate observable
\(\bar X = \ketbra{\bar 0}{\bar 1} + \ketbra{\bar 1}{\bar 0}\).
Hence, the information about \(\bar X\) is perfectly imprinted \(N\)
times in \(\mathcal S\), leading to an apparent paradox. Each \(\bar
X\) and \(\bar Z\) defines a set of redundantly imprinted observables,
yet each set is incompatible with the other. The measurement results
that can be gathered by observers measuring the redundant imprints
cannot be explained by resorting to a classical mixture of orthogonal
states. Here, redundancy is not enough to imply the classicality of
observables.

Nonetheless, it should be noted that both observables cannot be
measured simultaneously by different observers in spite of their
redundancy.  This is because any redundant record of \(\bar X\) and 
any reduncant record of \(\bar Z\) overlap in exactly one qubit and require incompatible
measurements for this specific qubit.  Thus, it is not possible to
share one redundant record of \(\bar Z\) with one observer, and one
about \(\bar X\) with another.  It is also not possible to have the
first observer perform a non-destructive measurement of \(\bar Z\) on
its part of \(\mathcal S\) and pass the overlapping qubit to the
second observer so that he/she measures \(\bar X\): the first measurement
already destroys the needed coherence for the second.

This remark is the core of the main result of \cite{R17:classical} for
recovering objectivity for quantum many-body systems. A sufficient 
criterion is introduced to guarantee that any two redundant records in
\(\mathcal S\), possibly corresponding to different observables \(F\)
and \(G\), can always be measured in any order and yet yield
compatible results. More precisely, it ensures that the state \(\ket
\psi\) of the whole system \(\mathcal S\) can be written as \(\ket
\psi = \sum_i \ket{\psi_i}\), where each \(\ket{\psi_i}\) is a
simultaneous eigenstate of \(F\) and \(G\), thereby ensuring the
orthogonality of the \(\ket{\psi_i}\) and the indistinguishability
between \(\ket \psi\) and \(\sum_i \ketbra{\psi_i}\) for feasible
measurements.

To make this formal (see~\cite{R17:classical} for details), suppose 
 \(\bm{F} = \{F_f\}_{f\in \mathcal
  F}\) and \(\bm G = \{G_g\}_{g\in\mathcal G}\) are two sets of
redundantly recorded observables on \(\mathcal S\) with respect to the
corresponding partitions \(\mathcal F\) and \(\mathcal G\) of the
microscopic sites \(\mathcal S_i\) of \(\mathcal S\).  This means that
for each element \(f\in \mathcal F\), there exists an observable \(F_f
\in \bm F\) on \(f\) that can be decomposed into projectors
\(\{F^\alpha_f\}_\alpha\) where \(\alpha\) is an eigenvalue of
\(F_f\) such that 
\begin{equation*}
\forall \alpha, \ \forall f' \in \mathcal{F}, \quad F^\alpha_f  \ket \psi = F^\alpha_{f'} \ket \psi,
\end{equation*}
and similarly for \(G_g\in\bm G\) on \(g \in \mathcal G\) with
projectors \(\{G^\mu_g\}_\mu\) associated to eigenvalues \(\mu\) of
\(G_g\).  Then, a sufficient condition on \(\mathcal F\) and
\(\mathcal G\) to ensure that results of \(F_f\) on \(f\) are
compatible with those of \(G_g\) on \(g\), for all values of $f$ and
$g$, is that for all \(f,f' \in \mathcal F\), there exists \(g \in
\mathcal G\), possibly depending on \(f\), and \(f'\) such that \(f
\cap g = f' \cap g = \emptyset\), and vice versa with the roles of
\(\mathcal F\) and \(\mathcal G\) permuted.  This property is called
\emph{non pair-covering} of \(\mathcal F\) and \(\mathcal
G\)~\cite{R17:classical}.

As a result, when $\mathcal F$ and $\mathcal G$ are not pair covering
each other, we have 
\begin{align*}
\forall f,f' \in \mathcal F,\ \exists g \in \mathcal G, \quad  F^\alpha_f G^\mu_g \ket \psi = F^\alpha_{f'} G^\mu_g \ket \psi, \\
\forall g,g' \in \mathcal G,\ \exists f \in \mathcal F, \quad  G^\mu_g F^\alpha_f \ket \psi = G^\mu_{g'} F^\alpha_f \ket \psi.
\end{align*}

In essence, this means that not only are there redundant imprints of
the observables in $\bm F$ in the state \(\ket \psi\) of \(\mathcal S\),
but the redundancy remains even though \(G_g \in \bm G\) is actively measured or
\(\ket\psi\) is decohered as a result of tracing out \(g \in \mathcal G\)
(and the same with the roles of \(\bm F\) and \(\bm G\) permuted).

This is indeed enough to impose the commutation on the support of
\(\ket \psi\): using the same notation, for any \(f\) and \(g\), the
non-pair covering condition gives
\begin{align*}
\exists g', & \quad f  \cap g' = \emptyset \\
\exists f', & \quad f' \cap g  = f' \cap g' = \emptyset.
\end{align*}

 {Then,} 
\begin{align}
F^\alpha_f G^\mu_g \ket \psi
& = F^\alpha_f G^\mu_{g'} \ket \psi \label{eq:commut_1} \\
& = F^\alpha_{f'} G^\mu_{g'} \ket\psi \label{eq:commut_2} \\
& = G^\mu_{g'} F^\alpha_{f'} \ket\psi \label{eq:commut_3} \\
& = G^\mu_{g} F^\alpha_{f'} \ket\psi \label{eq:commut_4} \\
& = G^\mu_{g} F^\alpha_{f}  \ket\psi \label{eq:commut_5},
\end{align}
where Equations (\ref{eq:commut_1}) and (\ref{eq:commut_5}) follow
from the redundancy of records,  Equations (\ref{eq:commut_2})
and (\ref{eq:commut_4}) derive from the non-pair covering condition,
and Equation (\ref{eq:commut_3}) is a direct consequence of the
absence of overlap between \(f'\) and \(g'\).

One can now prove by induction that the same holds for multiple sets
of redundantly imprinted observables \(\bm F, \bm G, \ldots \bm
Z\). Their projectors commute over \(\ket\psi\), allowing to define a
common branch decomposition for the state of the system as prescribed
by Equation (\ref{eq:common_branch}).

\section{Approximate Records and Classicality for Quantum Many-Body Systems \label{sec:generalize}}
\label{sec:orgcd3597f}
The significance of the non-pair covering criterion introduced in the
previous section is due to the relative ease with which it is met in
practice. The overlap that is required to maintain the ambiguity
between redundantly recorded, yet incompatible, observables is too
delicate to happen in realistic physical systems---see \cite{R17:classical} for an extended discussion on this point.

However, this reasoning suffers from several drawbacks. First, the
non-pair covering criterion is applicable only to (sharp) observables
and not to the broader information gathering strategies that can be
implemented using POVMs. Second, redundant observables must be
perfectly imprinted in fragments of \(\mathcal S\). Both restrictions
can be ultimately traced back to how redundancy is measured and how
classicality is deemed, that is, whenever projective measurements are compatible
on the state \(\ket \psi\) of the system, or equivalently, whenever they commute on
the support of \(\ket\psi\). 

The paragraphs below address these two points by providing a
definition of approximate redundant records of POVMs and an
alternative witness for their classicality.

\subsection{Approximate Copies of POVM Records}
\label{sec:orgc66d8cf}
Let \(\mathcal S = \cup_{i=1}^N \mathcal S_i\) be a many-body system
with \(N\) microscopic sites. Denote by \(\mathcal F\) a partition of
\([1,N]\) and by \(S_f = \cup_{i \in f} S_i\), for \(f\in
\mathcal F\).

\begin{definition}[\(\delta\)-approximate records]\label{def:perfect_imprint}
For $f, f' \in \mathcal F$ with $f\neq f'$, and two POVMs $F_f = \{F_f^\alpha\}_\alpha$ and $F_{f'}= \{F^\alpha_{f'}\}_\alpha$, respectively, on $S_f$ and $S_{f'}$.
For $\delta > 0$, we say that $F_{f'}$ \(\delta\)-approximately records $F_f$ on the system state $\rho$ if, $\forall \alpha$, 
\begin{equation*}
\tr(F^\alpha_f \otimes F^\alpha_{f'} \rho) \geq (1-\delta) \tr(F^\alpha_f \rho).
\end{equation*}
\end{definition}

As expected, this definition captures the fact that, given that outcome
\(\alpha\) is observed by measuring \(F_f\) on \(\rho\), a measurement of
\(F_{f'}\) yields the same outcome \(\alpha\) with a probability of at least
\(1-\delta\). This is because
\begin{equation*}
\Pr(F_{f'}\mbox{ yields outcome }\alpha | F_{f}\mbox{ yields outcome }\alpha) = \frac{\tr(F^\alpha_f \otimes F^\alpha_{f'} \rho)}{\tr(F^\alpha_f \rho)}.
\end{equation*}

When the above property is true for all \(f,f'\in\mathcal F\), we say
that the set of POVMs \(\bm =\{F_f\}_{f\in\mathcal F}\) is \emph{\(|\mathcal F|\)-times
\(\delta\)-approximately redundant}.

The following lemma shows that Definition \ref{def:perfect_imprint} falls back to that of \cite{R17:classical} for \(\delta = 0\), pure system states and (sharp) observables.
\begin{lemma}
Assume $F_f$ and $F_{f'}$ are projective measurements on disjoint subsets $f$ and $f'$ of $[1,N]$, and that $F_{f'}$ 0-approximately records $F_f$ on $\ket \psi$. Then
\begin{equation*}
F^\alpha_f \otimes F^\alpha_{f'} \ket\psi = F^\alpha_f \ket\psi.
\end{equation*}
\end{lemma}

\begin{proof}
Define the following normalized states
\begin{equation*}
\ket {\psi_{F^\alpha_f}} = \frac{F^\alpha_f \ket\psi}{\tr(F^\alpha_f \ketbra\psi)} \mbox{ and } \ket {\psi_{F^\alpha_{f'}}} = \frac{F^\alpha_{f'} \ket\psi}{\tr(F^\alpha_{f'} \ketbra\psi)}.
\end{equation*}

By assumption, $\tr(F^\alpha_f \otimes F^\alpha_{f'} \ketbra{\psi}) = \tr(F^\alpha_f \ketbra{\psi})$.
Using the definition of $\ket {\psi_{F^\alpha_f}}$, this becomes
\begin{equation*}
\tr(F^\alpha_{f'} \ketbra{\psi_{F^\alpha_f}})\times\tr(F^\alpha_f \ketbra\psi) = \tr(F^\alpha_f \ketbra\psi).
\end{equation*}

Hence, one concludes that $\tr(F^\alpha_{f'} \ket{\psi_{F_f^\alpha}}) = 1$, which implies that 
\begin{equation}\label{eq:partly_redundant}
F^\alpha_{f'} \ket{\psi_{F^\alpha_f}} = \ket{\psi_{F^\alpha_f}}.
\end{equation}

Similarly, for all $\alpha, \beta$, we have 
\begin{equation*}
\tr(F^\alpha_f \ketbra{\psi_{F^\beta_{f'}}}) = \frac{\tr(F^\alpha_f \otimes F^\beta_{f'}\ketbra\psi)}{\tr(F^\beta_{f'}\ketbra\psi)}.
\end{equation*}

Using equation (\ref{eq:partly_redundant}) on the rhs above and recalling that $F^\alpha_{f'} \times F^\beta_{f'} = F^\alpha_{f'} \times \one_{\alpha = \beta}$, we~obtain 
\begin{equation*}
\tr(F^\alpha_f \ketbra{\psi_{F^\beta_{f'}}}) = \frac{\tr(F^\alpha_f\ketbra{\psi})}{\tr(F^\beta_{f'}\ketbra{\psi})} \times \one_{\alpha = \beta}.
\end{equation*}

{For fixed $\beta$, taking the sum over $\alpha$ yields 1, because $\ket{\psi_{F^\beta_{f'}}}$ is normalized and $\sum_\alpha F^\alpha_f = \one$}, so that we can conclude that $\tr(F^\alpha_f \ket{\psi_{F_f^\alpha}}) = 1$.
In turn, this implies that $F^\alpha_f \ket{\psi_{F^\alpha_{f'}}} = \ket{\psi_{F^\alpha_{f'}}}$ and we arrive at
\begin{equation*}
F^\alpha_f \ket\psi = F^\alpha_f \otimes F^\alpha_{f'} \ket\psi = F^\alpha_{f'} \ket\psi.
\end{equation*}
\end{proof}

\subsection{Extending the Compatibility Criterion as a Witness for Classicality}
\label{sec:orgcd1ffc0}
As previously argued, one expects that quantum Darwinism for a
many-body system \(\mathcal S\) implies that (i) a preferred set of
POVMs emerges from the sole requirement of being approximately
redundantly recorded in the state of \(\mathcal S\), and (ii) these
POVMs exhibit classicality.

The natural choice of witness for classicality is that observers
accessing fragments of \(\mathcal S\) will be able to explain all the
correlations of their measurement results without the recourse to
quantum correlations. In \cite{R17:classical}, this is required for
arbitrary pure quantum states of the system, which translates into the
ability of the preferred observables to induce a decomposition of the
state \(\ket\psi\) of \(\mathcal S\) into a superposition of
orthogonal branches \(\ket \psi = \sum_i \ket{\psi_i}\), where each
\(\ket {\psi_i}\) is a common eigenstate of all observables in
redundantly imprinted sets $\bm O_1, \bm O_2, \ldots$, i.e.,
\begin{equation*}
\forall O \in \bm O_1 \cup \bm O_2 \cup \ldots, \ O \ket{\psi_i} = \omega(i,O) \ket{\psi_i},
\end{equation*}
thereby defining the compatibility of all the observables of $\bm O_1
\cup \bm O_2 \cup \ldots$ on \(\ket \psi\).

As anticipated, compatibility does not generalize straightforwardly to
POVMs due to the absence of a meaningful equivalent to eigenstates of
observables. Nonetheless, several options have been proposed in other
contexts to understand and sometimes quantify the classicality of
POVMs, namely through the introduction of commutativity,
non-disturbance, joint-measurability and coexistence (see, for example, \cite{HW10:nondisturbing, L83:foundations}). Our choice,
justified below, for the substitute for compatibility is based on
joint measurability.

\begin{definition}[Joint-measurability]
Let $\bm O$ be a set of POVMs, and for $O \in \bm O$ denote its
elements by $\{O^\omega\}_\omega$.  The set $\bm O$ is jointly
measurable if and only if there exists a POVM $T$ with elements 
$\{T^\theta\}_\theta$ such~that
\begin{equation}
\forall O \in \bm O, \ \forall \omega, \ O^\omega = \sum_\theta p(\omega | O, \theta) T^\theta, \label{eq:joint_meas}
\end{equation}
where $p(\omega | O, \theta)$ is a probability distribution for $\omega$ when $O$ and $\theta$ are fixed.
\end{definition}

This definition states that all measurements in \(\bm O\) can be
simulated by first measuring \(T\) and then, depending on the obtained
outcome \(\theta\) and the chosen $O \in \bm O$, by sampling
\(\omega\) according to the probability distribution \(p(\omega|O,
\theta)\).

This choice is motivated by the operational approach promoted by
quantum Darwinism. Observers can perform measurements, accumulate 
statistics and investigate correlations between them. When POVMs are
jointly measurable, observers are able to interpret the correlations
of measurement results through a simple marginalization process.

Joint measurability is further justified as a witness of classicality,
as it rules out steering---a purely quantum phenomenon---(see \cite{UCNG20:quantum} for a review). On the contrary, coexistence
can reveal steering \cite{UMG14:joint}, and is therefore not an
appropriate choice in our context. Additionally, non-disturbance
suffers from drawbacks in light of quantum Darwinism: it is
usually asymmetric, meaning that measurements need to be carried out
in a precise order so as to not disturb one another. This ordering
requirement contradicts our everyday experience of classical
features obviously robust to the precise order in which measurements
are performed. Finally, commutativity is shown to imply
joint measurability \cite{HW10:nondisturbing}, but the converse is in
general not true. Hence, without further good reasons to rule out
joint measurability, witnessing classicality through commutativity
risks being too restrictive and, thus, potentially missing the emergence
of objectivity.

Additionally, Proposition 1 of \cite{HW10:nondisturbing} shows that
when restricted to projective measurements, joint measurability is
indeed equivalent to the commutativity of observables. Thus, our choice of
witness for classicality reduces to that of
Ref. \cite{R17:classical}, as compatibility on the state of the
system reduces to commutativity on its support.

Lastly, to obtain a useful criterion for classicality in our context,
it needs to account for (i) approximations and (ii) systems whose
evolutions practically restrict their attainable states to a subset of
all possible density matrices.  To this end, we note that the operator
equality of Equation (\ref{eq:joint_meas}) is equivalent to a
statement on probabilities of the outcomes computed for system states 
\(\rho\) that span the set of density matrices for \(\mathcal S\). This is
because the trace function is an inner product for the real Hilbert
space of Hermitian matrices. Hence, we can deal with (i) by stating
that probability distributions are close to that obtained for
jointly measurable POVMs, and (ii) can be accounted for by enforcing
the relation only on the set \(\mathcal D\) of attainable states.

\begin{definition}[\(\delta\)-approximate joint measurability over $\mathcal D$]
Let $\mathcal D$ be a set of density matrices, $\delta \geq 0$ and
$\bm O$ a set of POVMs, where the elements of $O \in \bm O$ are
$\{O^\omega\}_\omega$.  The set $\bm O$ is \(\delta\) approximately
jointly measurable over $\mathcal D$ if there exists a POVM $T$ with
elements $\{T^\theta\}_\theta$ such that
\begin{equation}
\forall O \in \bm O, \ \forall \omega, \ \forall \rho \in \mathcal D, \ \left|\tr(O^\omega\rho) - \sum_\theta p(\omega | O, \theta) \tr(T^\theta\rho)\right| \leq \delta, \label{eq:approx_joint_meas}
\end{equation}
where $p(\omega | O, \theta)$ is a probability distribution for
$\omega$ when $O$ and $\theta$ are fixed.
\end{definition}

\section{Recovering Joint Measurability \label{sec:tuple}} 
\label{sec:org200e8c1}
As seen in Section \ref{sec:ideal}, redundancy is not enough to imply
classicality. The absence of a natural, or preferred, way to group
microscopic sites of a quantum many-body system allows information
about incompatible observables to be redundantly recorded in the whole
system. Although incompatible observables cannot be read off at the
same time by multiple observers---so that this statement does not
violate axioms of quantum mechanics---they can still collectively
decide beforehand which one to recover.

In the case of perfect redundant records of projective measurements,
the non pair-covering condition ensures that only a single set of
compatible observables can be accessed by observers, thus
corresponding to the everyday experience. Given our definitions of
approximate records and the replacement of compatibility with
approximate joint measurability, the question we have to address is
whether non pair-covering is enough to guarantee the joint measurability
of a single set of observables.

\begin{theorem}\label{thm:pair}
Let $\mathcal S$ be a quantum many-body system, such that there exists
$\mathcal F$, a partition of $[1,N]$ of the microscopic sites
$\mathcal S_i$ of $\mathcal S$. Let $\bm F = \{F_f\}_f$ be a set POVMs,
where $F_f$ acts on $f$ only and satisfies $\forall \alpha$ and
$\forall f,f' \in \mathcal F$,
\begin{equation*}
\forall \rho \in \mathcal D, \ \tr(F^\alpha_f \otimes F^\alpha_{f'} \rho) \geq (1-\delta) \tr(F^\alpha_f\rho),
\end{equation*}
for some $\delta > 0$, and $\mathcal D$ a set of density matrices.
Assume there exists $\mathcal G$, a second partition, and $\bm =
\{G_g\}_g$ with $g \in \mathcal G$ a second set of POVMs satisfying
the corresponding approximate redundantly recorded condition stated
above. Assume that $\mathcal F$ and $\mathcal G$ do not pair-cover
each other, then for all $f\in \mathcal F$ and $g \in \mathcal G$,
$F_f$ and $G_g$ are \(\delta\)-approximately jointly measurable on 
$\mathcal D$.
\end{theorem}

\begin{proof}
The non pair-covering condition imposes that
\begin{align*}
& \forall f,f'\in \mathcal F, \exists g \in \mathcal G,\mbox{ s.t. } f\cap g = \emptyset \mbox{ and } f'\cap g = \emptyset \\
& \forall g,g'\in \mathcal G, \exists f \in \mathcal F,\mbox{ s.t. } g\cap f = \emptyset \mbox{ and } g'\cap f = \emptyset.
\end{align*}

For given $f \in \mathcal F$ and $g \in \mathcal G$, using the non pair-covering condition, it is possible to choose $f' \in \mathcal F$ and $g' \in \mathcal G$ such that 
\begin{align*}
f \cap g' = \emptyset = f' \cap g' \\
g \cap f' = \emptyset = g' \cap f'.
\end{align*}

Then, using redundancy and the disjointness conditions above, for all $\alpha$, we obtain 
\begin{align*}
\tr(F^\alpha_f \rho)
& \geq \tr(F^\alpha_f\otimes F^\alpha_{f'} \rho) \\
& = \tr(F^\alpha_f\otimes F^\alpha_{f'} \otimes \sum_\nu G^\nu_{g'}\rho) \\
& \geq (1-\delta) \tr(F^\alpha_{f'} \otimes \sum_\nu G^\nu_{g'}\rho),
\end{align*}
and similarly for all $\nu$
\begin{align*}
\tr(G^\nu_{g} \rho)
& \geq \tr(G^\mu_{g} \otimes G^\mu_{g'} \rho) \\
& = \tr(G^\mu_{g} \otimes G^\mu_{g'} \otimes \sum_\beta F^\beta_{f'}\rho) \\
& \geq (1-\delta) \tr(G^\mu_{g'} \otimes \sum_\beta F^\beta_{f'}\rho).
\end{align*}

We also have
\begin{align*}
\tr (F^\alpha_f \rho)
& = 1 - \sum_{\beta\setminus\alpha}\tr(f^\beta_f\rho) \\
& \leq 1 - (1-\delta) \sum_{\beta\setminus\alpha}\tr(F^\beta_{f'} \otimes \sum_\nu G^\nu_{g'}\rho) \\
& = 1 - (1-\delta) (1-\tr(F^\alpha_{f'} \otimes \sum_\nu G^\nu_{g'}\rho)) \\
& = (1-\delta) \tr(F^\alpha_{f'} \otimes \sum_\nu G^\nu_{g'}\rho) + \delta,
\end{align*}
and similarly for $\tr(G^\nu_{g}\rho)$.

Combining both inequalities, we arrive at
\begin{align*}
& \forall \alpha, \nu, \ \left|\tr(F^\alpha_f \rho) - \tr(\sum_\nu F^\alpha_{f'} \otimes G^\nu_{g'} \rho) \right| \leq \delta, \mbox{ and} \\
& \forall \beta, \mu, \ \left|\tr(G^\mu_f \rho) - \tr(\sum_\beta F^\beta_{f'} \otimes G^\mu_{g'} \rho) \right| \leq \delta.
\end{align*}

This concludes the proof, as the probabilities of obtaining outcomes $F^\alpha_f$ and $G^\mu_{g}$ are \(\delta\)-close to that obtained by measuring $F^\alpha_{f'} \otimes G^\mu_{g'}$ followed by the appropriate post processing, consisting of summing over the outcomes of the ignored POVM. 
\end{proof}

Hence, any pair of approximately redundantly recorded POVMs is 
approximately jointly measurable. The trouble to recover a perfect
analogue to the ideal case with pure states and projective
measurements is that pairwise joint measurability does not imply
global joint measurability \cite{HFR14:quantum}. That is, for three
POVMs, all pairs can be jointly measurable, but all three of them
might not be the marginals of a single POVM. As a consequence, one
cannot claim full classicality in such a situation.

Global joint measurability can nonetheless be obtained by
strengthening the non pair-covering condition into non tuple-covering.
\begin{definition}[non tuple-covering]
$\mathcal F, \mathcal G, \ldots, \mathcal Z$ partitions of $[1,N]$ are non tuple-covering each other iff,
$\forall f \in \mathcal F, \ g \in \mathcal G, \ldots, \ z \in \mathcal Z, \ \exists f' \in \mathcal F, \ g' \in \mathcal G, \ldots, \ z' \in \mathcal Z$ s.t.
\begin{align*}
f' \cap g & = f' \cap g' = \ldots = f' \cap z = f' \cap z' = \emptyset \\
g' \cap f & = g' \cap f' = \ldots = g' \cap z = g' \cap z' = \emptyset \\
& \vdots \\
z' \cap f & = z' \cap f' = z' \cap g = z' \cap g' = \ldots = \emptyset. 
\end{align*}
\end{definition}

{Using this definition, the following theorem allows to recover global joint measurability.}
\begin{theorem}
Let $\bm F = \{F_f\}_{f\in \mathcal F}, \ \bm G =\{G_g\}_{g\in \mathcal G}, \ldots \ \bm Z = \{Z_z\}_{z \in \mathcal Z}$ be sets of $\delta$-approximate redundantly recorded POVMs on the state of a quantum many-body system $\mathcal S$, with $\mathcal F, \mathcal G, \ldots \mathcal Z$ partitions of $[1,N]$, the indices of the microscopic sites. 
If the partitions $\mathcal F, \mathcal G, \ldots \mathcal Z$ do not tuple-cover each other, then for any $f,g,\ldots z$, $F_f, G_g, \ldots Z_z$ are \(\delta\)-approximately joint measurable.  
\end{theorem}

\begin{proof}
Given the  non tuple-covering condition, one could appropriately replace any measurement of $F_f, G_g, \ldots Z_z$ by a measurement of $F_{f'}, G_{g'}, \ldots Z_{z'}$. 
From there, the same proof technique as the one used for Theorem \ref{thm:pair} applies.
Using the said replacement of measurements, one arrives at a situation where all POVMs $F_{f'}, G_{g'}, \ldots, Z_{z'}$ act on different subsets of the microscopic sites.
They are, thus, defining a global POVM with elements $F^\alpha_{f'}\otimes G^\beta_{g'}\otimes \ldots Z^\zeta_{z'}$ from which the probabilities of the outcomes $(\alpha, \beta, \ldots, \zeta)$ can be \(\delta\)-approximated through classical post-processing.
This allows to conclude about the \(\delta\)-approximate joint-measurability criterion for POVMs $\{F_f\}_{f\in \mathcal F}, \{G_g\}_{g\in \mathcal G}, \ldots, \{Z_z\}_{z \in \mathcal Z}$.
\end{proof}

\section{Dynamical Approach to the Emergence of Classicality \label{sec:dynamical}}
\label{sec:org9419e64}
The non pair-covering condition has an appealing property of being
rather simple and allowing the recovery of objectivity for usual
many-body physics experiments: pair-covering is too delicate to
maintain for macroscopic systems containing possibly millions or
billions of microscopic sites so that they would necessarily be
exhibiting only usual classical properties.

On the contrary, the non tuple-covering seems a more complex, if not
harder, condition to achieve. This, in turn, weakens considerably the
above argument and, as a consequence, the reach of quantum Darwinism
for quantum many-body systems. Yet, we prove below that this is not
the case, and that quantum Darwinism is a ubiquitous mechanism to
explain the emergence of a single set of approximately
jointly-measurable POVMs.

The way to address this question is to take a dynamical view at the
creation of the redundant imprints into the state of the quantum
many-body system. More precisely, we need to acknowledge the fact that
the redundant imprints---be they perfect or \mbox{approximate---are} the
result of an evolution from some initial state of an initial
uncorrelated system \(\mathcal R\). In other terms, it results from the
transformation of a state \(\sigma \in \mathcal D(\mathcal R)\) to a
state \(\rho \in \mathcal D(\mathcal S)\), where \(\mathcal D(\mathcal
R)\) is the set of density matrices for \(\mathcal R\) and
similarly for \(\mathcal S\). The transformation can then be
represented by a CPTP map \(\Lambda\) so that \(\rho =
\Lambda(\sigma)\).

The structure of the correlations, and hence of the information,
between \(\mathcal R\) and \(\mathcal S\) can be analyzed using the
techniques pioneered in \cite{BPH15:generic} and refined
in \cite{QR21:emergent}. Yet, these need to be recast to fit into the
quantum many-body setting, as they have been developed in the
system--environment context.

\begin{theorem}\label{thm:ranard}
  Let $\Lambda$ be a CPTP map from $\mathcal D(\mathcal R)$ to
  $\mathcal D(\mathcal S)$, and $w_q, w_f \in [1,N]$, with $\mathcal
  S=\cup_{i=1}^N \mathcal S_i$ and $w_q + w_f \leq N$. For all $\sigma
  \in \mathcal D(\mathcal R)$, consider $\varrho = \Lambda(\sigma)$
  the state of a generic quantum many-body system that evolved from
  the initial preparation state $\sigma$ through $\Lambda$.  Then,
  there exists a subset $q$ of $[1,N]$ of size at most $w_q$ such that
  for all subsets $f$ of $[1,N]\setminus q$ with size $w_f$, and for
  all POVMs $F_f = \{F^\alpha_f\}_\alpha$ on $f$
  \begin{equation*}
    \forall \alpha, \ \left| \tr(F^\alpha_f \varrho) - \sum_\theta p(\alpha | F_f, \theta) \tr( T^\theta_q \varrho) \right|
    \leq \delta,
  \end{equation*}
  with $\delta = d_{\mathcal R} \sqrt{2\ln(d_{\mathcal R})
    \frac{w_f}{w_q}}$, where $T_q$ is a fixed POVM on $q$ that does
  not depend on $f$ nor $\sigma$, and where $p(\alpha | F_f, \theta)$
  is a classical probability distribution for $\alpha$ when $F_f$ and
  $\theta$ are fixed that is independent of $\sigma$. Above,
  $d_{\mathcal R}$ denotes the dimension of $\mathcal R$.
\end{theorem}

\begin{proof}
  The proof will proceed in two steps. First, it will follow the steps
  of Theorem~2 of~\cite{QR21:emergent} to obtain a bound on the
  distance between the Choi-states of two specific channels, one being
  the channel $\Lambda$ reduced to some sufficiently small subsets $f$
  and the other one being a measure and prepare channel from $\mathcal
  R$ to $f$. The second step will focus on the measurement done by the
  measure and prepare channel and show that it can be understood as a
  measurement on a subset $q$ disjoint and independent of $f$.
  
  Consider a basis $\ket i$ of $\mathcal R$ and a fiducial reference
  system $\mathcal R'$ isomorphic to $\mathcal R$.  Define the
  maximally mixed state $\ket \psi$ of $\mathcal R\mathcal R'$ as
  $1/\sqrt{d_\mathcal R} \sum_i \ket {ii}_{\mathcal R\mathcal R'}$.
  The Choi-state of $\Lambda$ is then $\rho = (\one_{\mathcal
    R'}\otimes\Lambda(\ketbra\psi)$ (see, for example,~\cite{W17:quantum}).  
 We
  can now apply Proposition~1 of~\cite{QR21:emergent} to $\rho$.  For
  $w_f, w_q \in [1,N]$, there exists $q \subseteq [1,N]$ of size $w_q$
  and $\Xi_q$ a quantum--classical channel on $q$ such that
  \begin{equation*}\label{eq:prop1}
    \forall f \subseteq [1,N] \setminus q, \ |f| = w_f, \ \max_{\Xi_f \in QC} I(\mathcal R' : f |q)_{\Xi_f \otimes \Xi_q(\rho)} \leq S(\mathcal R')_{\rho} \frac{w_f}{w_q}.
  \end{equation*}

  Above, $\Xi_f$ is a quantum--classical channel on $f$ such that
  $\Xi_f(X) = \sum_\alpha \tr(F^\alpha_f \ketbra{\alpha})$ for some
  POVM $F_f$ on $f$; $S(\mathcal R')_{\rho}$ is the von Neumann
  entropy for the system $\mathcal R'$ when the global state is
  $\rho$; and $I(\mathcal R':f|q)_{\Xi_f \otimes \Xi_q(\rho)}$ is the
  quantum mutual information between $\mathcal R'$ and $f$ conditioned
  on $q$ for the global state $\Xi_f \otimes \Xi_q (\rho)$---note
  that to ease notation, the obvious identity operators will continue
  to be omitted.  The interest of this proposition is that it
  constructs a subset $q$ of microscopic subsystems of $\mathcal S$ of
  size at most $w_q$ such that, \emph{irrespective} of the choice of
  another subset $f$ of microscopic subsystems of size $w_f$ disjoint
  from $q$, the correlations between $\mathcal R'$ and any observation
  on $f$ through $\Xi_f$ conditioned on an observation of $q$ through
  $\Xi_q$ can be made small.  This means that observing $q$ through
  $\Xi_q$ extracts all there is to know about $\mathcal R'$ so that it
  becomes uncorrelated with any further observation on $f$.  By
  analogy with the classical case, {Ref.} 
 \cite{QR21:emergent} refers to the
  region $q$ as a quantum Markov blanket.  We can now use this bound
  to arrive at a statement of closeness between two Choi-states.  More
  precisely, for $\Xi_q$ implementing the POVM $T_q =
  \{T^\theta_q\}_\theta$ on $q$ so that $\Xi_q(X) = \sum_\theta
  \tr(T^\theta_q X) \ketbra{\theta}$, we have
  \begin{align}
    \tr_{\bar f} (\Xi_f \otimes \Xi_q (\rho)) & =  \Xi_f \left(\sum_{\theta} p_\theta \rho^\theta_{\mathcal R' f} \ketbra{\theta}\right), \mbox{ with } \\
    p_\theta & =  \tr (T^\theta_q \rho) \\
    \rho^\theta_{\mathcal R'f} & =  \frac{1}{p_\theta}\tr_{\bar f q}(T^\theta_q \rho) \label{eq:meas_q},
  \end{align}
  where $\bar f$ is the complement of $f$ in $[1,N]\setminus q$ so
  that the system $\mathcal S$ decomposes into $f \bar f q$.  As a 
  consequence, $I(\mathcal R' : f | q)_{\Xi_f \otimes \Xi_q(\rho)} =
  \sum_\theta I(\mathcal R' : f)_{\Xi_f(\rho^\theta_{\mathcal R'f})}$.
  Using the quantum Pinsker inequality~\cite{HOT81:kms} for $I(\mathcal R' :
  f)_{\Xi_f(\rho^\theta_{\mathcal R'f})}$, one obtains that   
  \begin{equation*}
    \frac{1}{2 \ln 2} \left\| \Xi_f (\rho^\theta_{\mathcal R'f} - \rho^\theta_{\mathcal R'}\otimes \rho^\theta_f)\right\|_1^2 \leq I(\mathcal R' : f)_{\Xi_f( \rho^\theta_{\mathcal R'f})}.
  \end{equation*}
  
  This being true for all $\theta$, using the convexity of both the
  square function and the 1-norm, we obtain
  \begin{equation*}
    \frac{1}{2 \ln 2} \left\| \Xi_f (\rho_{\mathcal R'f} - \sum_\theta p_\theta \rho^\theta_{\mathcal R'}\otimes \rho^\theta_f)\right\|^2_1 \leq I(\mathcal R' : f | q)_{\Xi_f \otimes \Xi_q (\rho)}.
  \end{equation*}
  
  Now, using Equation~(\ref{eq:prop1}) and $S(\mathcal R') \leq \log(d_\mathcal R)$, we have that for all quantum--classical channels $\Xi_f$ on $f$:
  \begin{equation}\label{eq:bound_choi}
  \left\| \Xi_f (\rho_{\mathcal R'f} - \sum_\theta p_\theta \rho^\theta_{\mathcal R'}\otimes \rho^\theta_f)\right\|_1 \leq \sqrt{2\ln(d_\mathcal R) \frac{w_f}{w_q}}.
  \end{equation}
  
  Above $\rho_{\mathcal R'f}$ is the Choi-state corresponding to
  $\Lambda_f$ obtained by reducing the channel $\Lambda$ to $f$, while
  $\sum_\theta p_\theta \rho^\theta_{\mathcal R'} \otimes
  \rho^\theta_f$ defines $\Gamma_f$, corresponding to a measure and
  prepare channel from $\mathcal R$ to $f$, as its Choi-state is
  separable with respect to the $\mathcal R' f$ partition.  Note that
  in $\Gamma_f$, the prepared states $\rho^\theta_f$ are independent
  of the input of the channel.

  We can define two additional channels, $\Lambda^{\Xi_f}_f =
  \Xi_f\circ\Lambda_f$ and $\Gamma^{\Xi_f}_f = \Xi_f\circ \Gamma_f$
  that correspond to the Choi-states $\Xi_f(\rho_{\mathcal R'f})$ and
  $\Xi_f(\sum_\theta p_\theta \rho^\theta_{\mathcal R'}\otimes
  \rho^\theta_f)$, respectively, and that are realized by measuring the
  output states of $\Lambda_f$ and $\Gamma_f$ with the POVM
  $F_f= \{F^\alpha_f\}_\alpha$. Then, we have, for all $\Xi_f \in QC$
  \begin{equation*}
    \left\|\Lambda^{\Xi_f}_f - \Gamma^{\Xi_f}_f \right\|_\diamond \leq d_{\mathcal R} \left\| \Xi_f (\rho_{\mathcal R'f} - \sum_\theta p_\theta \rho^\theta_{\mathcal R'}\otimes \rho^\theta_f)\right\|_1
  \end{equation*}
  which implies, as the diamond norm is the result of an optimization over all input states and because of Equation~(\ref{eq:bound_choi}), that
  \begin{equation}
    \forall \sigma \in \mathcal D(\mathcal R), \ \left\|\Lambda^{\Xi_f}_f(\sigma) - \Gamma^{\Xi_f}_f(\sigma)\right\|_1 \leq d_{\mathcal R} \sqrt{2\ln(d_\mathcal R) \frac{w_f}{w_q}}.\label{eq:almost}
  \end{equation}

  We almost arrive at our result and just need to give a more
  explicit interpretation to both states in the above equation.
  $\Lambda^{\Xi_f}_f(\sigma) = \sum_{\alpha}\tr(F^\alpha_f
  \Lambda(\sigma))\ketbra{\alpha}$ is the state obtained after
  measuring $\Lambda(\sigma)$ using $F_f$
  acting on subset $f$ of size $w_f$. To interpret the state
  $\Gamma^{\Xi_f}_f(\sigma)$, recall that the output of a given
  channel $\Phi$ from $\mathcal R$ to $f$ can be inferred from its
  corresponding Choi-state $\rho^\Phi_{\mathcal R'f}$, using the simple
  identity $\Phi(\sigma) = \tr_{\mathcal R'} (\rho^{\Phi}_{\mathcal
    R'f}\sigma^T )$.  Therefore, we have
  \begin{align*}
    \Gamma^{\Xi_f}_f(\sigma)
    & = \tr_{\mathcal R'}\left(\sum_\theta p_\theta \rho^\theta_{\mathcal R'} \sigma^T\right) \otimes \left(\sum_\alpha \tr(F^\alpha_f \rho^\theta_f) \ketbra{\alpha}\right) \\
    & = \tr_{\mathcal R'}\left(\sum_\theta \tr_{f \bar f q} (T^\theta_q \rho) \sigma^T \right) \otimes \left(\sum_\alpha \tr(F^\alpha_f \rho^\theta_f) \ketbra{\alpha}\right) \\
    & = \tr \left(\sum_\theta T^\theta_q \Lambda(\sigma)\right) \otimes \left(\sum_\alpha \tr(F^\alpha_f \rho^\theta_f) \ketbra{\alpha}\right),
  \end{align*}
  where we use Equation~(\ref{eq:meas_q}) to replace $p_\theta
  \rho^\theta_{\mathcal R'}$ with $\tr_{f \bar f q} T^\theta_q
  \Lambda(\rho)$.  Note that for the states $\rho^\theta_f$ for varying $\theta$ are independent of 
  $\sigma$ so that $\tr(F^\alpha_f \rho^\theta_f)$ can be rewritten as 
  $p(\alpha| F_f, \theta)$, a classical probability distribution for 
  $\alpha$, given $F_f$ and $\theta$. Equation~(\ref{eq:almost}) can now be rewritten as 
  \begin{equation*}
    \forall \sigma \in \mathcal D(\mathcal R), \ \sum_{\alpha} \left\|
    \tr(F^\alpha_f \Lambda(\sigma))
    - \tr \left(\sum_\theta T^\theta_q \Lambda(\sigma)\right) \tr(F^\alpha_f \rho^\theta_f)
    \right\|_1 \leq d_{\mathcal R} \sqrt{2\ln(d_\mathcal R) \frac{w_f}{w_q}}.
  \end{equation*}

  All derivations above are independent from the choice of subset $f$
  and of quantum-classical channel $\Xi_f$---or, equivalently, of $F_f$---as long as $w_f$ and $w_q$ are chosen such that $\delta =
  d_{\mathcal R} \sqrt{2\ln(d_\mathcal R) \frac{w_f}{w_q}}$ is
  small. This concludes the proof as
\begingroup\makeatletter\def\f@size{9.5}\check@mathfonts
\def\maketag@@@#1{\hbox{\m@th\normalsize\normalfont#1}}%
  \begin{equation*}
    \forall \varrho \in \mathcal D, \ \forall \alpha, \ \left| \tr(F^\alpha_f \varrho) - \sum_\theta p(\alpha | F_f, \theta) \tr( T^\theta_q \varrho) \right|
    \leq \sum_{\tilde \alpha} \left| \tr(F^{\tilde\alpha}_f \varrho) - \sum_\theta p(\tilde \alpha| F_f, \theta) \tr( T^\theta_q \varrho) \right|
    \leq \delta.
  \end{equation*}
  \endgroup
\end{proof}

In effect, Proposition~1 of~\cite{QR21:emergent} identifies a fraction
of $\mathcal S$ that contains all the information that can be accessed
about the initial state $\sigma$ after $\Lambda$ has taken place. This
then decoheres all other possible smaller fractions $f$ of $\mathcal
S$ disjoint from $q$. The consequence is that any measurement on such
fractions can be implemented by first measuring $q$ and then by
post processing classically the result depending on the choice of
measurement $F_f$ on $f$. This being true for all sufficiently small
fractions $f$ and any measurement on $F_f$, we recover the
$\delta$-approximate joint measurability for such measurements over
the states that are dynamically created by $\Lambda$ from any initial 
state $\sigma$. Hence, when observers are restricted to fractions $f$,
quantum Darwinism yields objective properties of the system that can
all be understood as stemming from a single classical measurement on
the Markov blanket $q$.

\section{Conclusions}
\label{sec:org779a086}
The last section shows that generic evolutions of quantum many-body
systems do systematically generate Markov blankets that capture all
correlations between fragments of \(\mathcal S\). As a consequence,
measurement results obtained by observers measuring fragments of
\(\mathcal S\) outside Markov blankets can be explained using classical
correlations only. This implies that the non tuple-covering condition
is generically satisfied for all partitions of \(\mathcal S\) that
contain the Markov blanket. Hence, while the non tuple-covering
condition seemed an {a priori} more complex requirement to satisfy
compared to the non pair-covering, as soon as Markov blankets are
outside the reach of observers, quantum Darwinism can be invoked to
recover robust classical objective properties of quantum many-body
systems. This is a situation similar to that of system--environment
settings, where Markov blankets are created generically by quantum
evolutions and are responsible for objective classical
\mbox{reality \cite{BPH15:generic, QR21:emergent}.}  Further analysis of the
precise location and accessibility of Markov blankets in realistic
settings is left for future work.

\vspace{6pt}

\paragraph{Aknowledgements.}{I am extremely grateful to Wojciech Zurek for having introduced me to quantum information, for his guidance, stimulating discussions and the wonderful time at the Los Alamos National Laboratory.}

\bibliographystyle{alpha}
\bibliography{../qubib/qubib}

\end{document}